\documentclass[onecolumn]{article}
\usepackage{amsmath,amssymb,amsthm}
\usepackage{graphicx} 
\usepackage{fullpage}

\theoremstyle{definition}

\newtheorem{fact}{Fact}

\newcommand{\set}[1]{\{#1\}}

\newcommand{\sd}{\,|\,} 
\newcommand{\struct}[1]{\mathcal{#1}} 

\title{
  Corrections to the results derived in ``A Unified Approach to Algorithms Generating 
Unrestricted and Restricted Integer Compositions
and Integer Partitions''; and a comparison of four restricted integer composition generation algorithms 
} 
\author{Steffen Eger}

\begin{document}
\maketitle
\begin{abstract}
In this note, 
I discuss results on integer compositions/partitions given in the paper ``A Unified Approach to Algorithms Generating
Unrestricted and Restricted Integer Compositions
and Integer Partitions''. I also experiment with four different
generation algorithms for restricted integer compositions and find the
algorithm designed in the named paper to be pretty slow, comparatively.\\
 Some of my comments may be subjective.
\end{abstract}

\section{Introduction}
A few years ago, I became interested in the subject of
\emph{restricted integer compositions} because they arise in the
context of generalized sequence alignment algorithms. In particular, I
wanted a possibly fast algorithm for generating all compositions of an
integer $n$ with $k$ parts, each in the discrete interval
$\set{a,a+1,\ldots,b}$. Googling, I found both a Matlab implementation
and a paper reference, \cite{opdyke}. In its abstract, the algorithm
was praised as ``reasonably
fast with good time complexity'' and also as solving the ``open problem
of counting the number of integer compositions doubly restricted in
this manner''. 
After some experimentation, however, it appeared to me that 
not only is the 
discussed 
algorithm slow 
but, in
addition, the paper makes many (at best) misleading 
statements 
concerning (mathematical) results on restricted
integer compositions. 

In this note, I outline my objections to the Opdyke algorithm
and the paper's theoretical results and explain why I think the
algorithm and its underlying results are 
problematic. 
I first briefly 
summarize my points of critique, before I
introduce notation and definitions and, subsequently, detail my
concerns. Finally, I run the algorithm and compare it with
other (simple and not so simple) algorithms for generating restricted
(and unrestricted) integer compositions, where each part lies in an
arbitrary interval as outlined above. These experiments reveal that,
in fact, the algorithm's run time appears to be pretty bad,
exponentially worse (in one of the input parameters) than a competitor
algorithm's designed in \cite{vajn}. 

\begin{enumerate}
  \item Paper \cite{opdyke} claims to provide ``closed form solution[s]
    to the open problem of counting the number of [doubly
      restricted] integer compositions''. I argue that, on the
    contrary, the recursions that the paper indicates as closed-form
    solutions are (1) elementary, (2) long known, and (3) special
    cases of results developed in the literature paper 
    \cite{opdyke} cites. 
  \item Paper \cite{opdyke} claims to generalize earlier approaches to
    the restricted integer composition/partition problem by allowing
    \underline{both lower and upper bounds on the \textbf{value} of parts} in integer
    compositions/partitions. I argue that, on the contrary, general
    lower and upper bounds, $a$ and $b$, may, without loss of
    generality, be reduced to the special case $a=0$. Thus, in this
    respect, paper \cite{opdyke} is just as specific as other work in
    this field.
  \item Paper \cite{opdyke} claims to generalize earlier approaches to
    the restricted integer composition/partition problem by allowing
    \underline{both lower and upper bounds on the \textbf{number} of
      parts} in integer 
    compositions/partitions. I argue that, on the contrary, this is
    merely a trick to feign generality. What the paper does is simply
    to provide a \emph{wrapper function} (sic!) around its integer
    composition/partition generation module, invoking it with
    different parameter values for the number of parts.
  \item Paper \cite{opdyke} claims that its outlined algorithm is,
    ``given its generality, [...] reasonably fast with good time
    complexity''. I argue that, on the contrary, the paper's
    methodology is not general at all (as outlined), and hence, its
    time complexity is bad. 
  \item Paper \cite{opdyke} claims that it ``unifies'' the generation
    approach to the integer composition/partition problem. I argue
    that such a `unification' need hardly be surprising given that
    compositions and partitions are so closely related. I illustrate by
    providing other links between compositions and partitions. 
\end{enumerate}

\bigskip

An \emph{integer composition} of a nonnegative integer $n$ is a tuple
$(\pi_1,\dotsc,\pi_k)$ of nonnegative integers such that
$n=\pi_1+\cdots+\pi_k$.\footnote{Sometimes, the literature
  distinguishes between \emph{weak compositions} and
  \emph{compositions} but I will not do so. 
  I
  consider a generalized concept of integer compositions where parts
  may lie in some arbitrary subset of the nonnegative integers (one
  could allow for the whole of $\mathbb{Z}$).} The $\pi_i$'s are
usually 
called the 
\emph{parts} of the composition.
We 
call an integer composition \emph{$A$-restricted}, for a subset $A$ of the
nonnegative integers, if each part lies in $A$. If $\pi_1\ge
\pi_2\ge\cdots\ge \pi_k$, then $(\pi_1,\dotsc,\pi_k)$ is called an
\emph{integer partition}. We denote by $\struct{C}_{A}(n,k)$ the
set of restricted integer compostions of $n$ with fixed number $k$ of
parts, each in the set $A$. Analogously, we denote by
$\struct{P}_A(n,k)$ the set of restricted integer partitions of $n$
with fixed number $k$ of parts, each in the set $A$. By $c_A(n,k)$ and
$p_A(n,k)$, we denote the respective cardinalities. Throughtout, we
typically consider $A=[a,b]=\set{a,a+1,\dotsc,b}$. In the latter
cases, we also write $c(n,k,a,b)$ and $p(n,k,a,b)$, respectively. 

\bigskip

I now address several issues discussed in paper \cite{opdyke}.

\begin{enumerate}
  \item \textbf{The counting formulae}: In the abstract of the paper,
    one reads (bold added by myself)
    \begin{quote}
      A general, closed form solution to the \textbf{open problem
      of counting the number of integer compositions doubly restricted
      in this manner also is presented}; [...]
\end{quote}
    and, on p.67,
\begin{quote}
Formulae (2), (3), and (4) mirror the analogous solutions for counting
doubly restricted 
integer partitions presented later in the paper. Although their recursive
nature makes these formulae less convenient than, say, a simple combinatoric
equation or sum, \textbf{they still provide closed form solutions to
  problems which had 
none before}, and their calculation is not onerous.
\end{quote}
On p.77, the paper continues
\begin{quote}
  [...] which is probably why their important link to the \textbf{completely
  original, analogous solutions of (2), (3), and (4) for compositions
  has been missed until now}. 
\end{quote}

Denoting by $c(n,a,b)$ the number
of integer compositions of $n$ with arbitrary number of parts, each between
$a$ and $b$, and 
by $c(n,k_0\le k\le k_1,a,b)$ the
number of integer compositions of $n$ with $k$ parts, for $k_0\le k\le
k_1$, each between $a$ and $b$, 
\textbf{these closed form solutions are} (we shift equation numbers to match
those of the paper in question):
\addtocounter{equation}{1}
\begin{align}
  \label{eq:1}
  c(n,a,b) &= I(n\le b) + \sum_{i=\max\set{1,n-b}}^{n-a}c(i,a,b),\\
  \label{eq:2}
  c(n,k,a,b) &= \sum_{i=\max\set{1,n-b}}^{n-a}c(i,k-1,a,b),\\
  \label{eq:3}
  c(n,k_0\le k\le
  k_1,a,b)&=\sum_{k=k_0}^{k_1}\sum_{i=\max\set{1,n-b}}^{n-a}c(i,k-1,a,b). 
\end{align}
Here, $I(\cdot)$ denotes the indicator function, which is $1$ or $0$,
depending on whether the expression in brackets is true or not. 
\textbf{I argue that these three results are (1)
\emph{elementary}, 
(2) \emph{long known}, 
and (3) \emph{given in the 
  references the 
paper under scrutiny, \cite{opdyke}, cites}.} 

First, for Equation \eqref{eq:3}, there is nothing to prove since this
formula is, by definition, just the sum, over the number of parts, of
the formula
given in Equation \eqref{eq:2}. 
To prove results \eqref{eq:1} and \eqref{eq:2}, 
note that any
$A$-restricted integer composition of $n$ is obtained by adding $x$ to
a composition of $n-x$, for $x\in A$. In other words, denoting by
$c_A(n)$ the number of $A$-restricted integer compositions, 
we have:
\begin{align}
  \label{eq:cA}
  c_A(n) &= \sum_{x\in A}c_A(n-x),\\
  \label{eq:cAk}
  c_A(n,k) &= \sum_{x\in A}c_A(n-x,k-1).
\end{align}
Hence, if we specialize to
$A=\set{a,a+1,\ldots,b}$, Equations \eqref{eq:1} and \eqref{eq:2} are
retrieved. This shows that all three equations are elementary. 

To show that the equations are long known, Equation \eqref{eq:2}
is, for example, given in \cite{abramson}, Formula (5.4), and Equation
\eqref{eq:1} is, for example, given in \cite{hogatt}, Formula (4.6).

To show that Equations \eqref{eq:1} and \eqref{eq:2} are given in the
references of paper \cite{opdyke}, note that \eqref{eq:1} is given in
\cite{kimberling}, Lemma 3.1, and \eqref{eq:2} is given in
\cite{heubach}, proof of Theorem (2.1).  

Concerning the counting formulae for restricted integer partitions, I
admit that I 
am not familiar with the literature on integer partitions. However,
the formulae that paper \cite{opdyke} derives 
probably are given in any work on the topic (the author cites
\cite{andrews} as a reference). Namely, the formulae
are
\begin{align}
  \label{eq:7}
  p(n,a,b) &= I(n\le b) + \sum_{i=\max\set{1,n-b}}^{n-a}c(i,n-i,b),\\
  \label{eq:8}
  p(n,k,a,b) &= \sum_{i=\max\set{1,n-b}}^{n-a}p(i,k-1,n-i,b)
\end{align}
(I omit the formula that sums over different parts because it is
trivial). 
Deriving \eqref{eq:8} (and \eqref{eq:7}) is also simple. Each
partition must end, in its final part, with a number $x$, for some
$x\in A$. Since $x$ is (weakly) the smallest part of the partition,
the remaining $k-1$ 
parts must have size at least $x$ and they must sum to $n-x$. Hence,
\begin{align*}
  p_A(n,k) = \sum_{x\in A}p_{A_x}(n-x,k-1), 
\end{align*}
where $A_x=\set{y\in A\sd y\ge x}$. This generalizes formula
\eqref{eq:8}. Formula \eqref{eq:7} is completely analogous.

\item \textbf{Lower \textsc{and} upper bound restrictions?} Is it
  necessary to consider both \textbf{lower and upper bounds} in integer
  compositions and partitions? A well-known fact of restricted
  compositions and partitions is the
  following (see \cite{kimberling}, who states this
  as well-known, without proof, only for compositions; however, partitions
  are of course 
  analogous in this respect).
  
  \begin{fact}\label{fact:1}
    There exists a bijection $f$ between $\struct{C}_{[a,b]}(n,k)$ and
    $\struct{C}_{[0,b-a]}(n-ka,k)$, and there exists a bijection $g$
    between $\struct{P}_{[a,b]}(n,k)$ and $\struct{P}_{[0,b-a]}(n-ka,k)$.
  \end{fact}
  \begin{proof}
    Let
    $\pi=(\pi_1,\ldots,\pi_k)\in \struct{P}_{[a,b]}(n,k)$. Let 
    \begin{align*}
      g(\pi) = (\pi_1-a,\ldots,\pi_k-a).
    \end{align*}
    Of course, $g(\pi)\in \struct{P}_{[0,b-a]}(n-ka,k)$. Also, if
    $\pi\neq \pi'$, then clearly $g(\pi)\neq g(\pi')$. Finally, let
    $\tau=(\tau_1,\ldots,\tau_k)$ be any element from
    $\struct{P}_{[0,b-a]}(n-ka,k)$.  Then,
    $\tau'=(\tau_1+a,\ldots,\tau_k+a)\in \struct{P}_{[a,b]}(n,k)$ and
    $g(\tau')=\tau$. Hence, $g$ is injective and surjective, and
    consequently also bijective. 

    Since order didn't matter for our argument, the same conclusion
    holds for compositions. 
  \end{proof}
  Fact \ref{fact:1} states that --- at least from a mathematical
  perspective --- it suffices to consider the restricted
  integer/partition problem \emph{only with upper bounds} and lower
  bound $a=0$. From a computational perspective, if an algorithm
  $A(n,k,0,c)$ is given which generates all compositions/partitions of
  $n$ with 
  $k$ parts, each between $0$ and $c$, and which takes time $O(k)$ per
  composition (as the Opdyke algorithm claims it does), then there always also
  exists an $O(k)$ algorithm which generates all
  compositions/partitions of $n$ 
  with $k$ parts, each between some lower bound $a\ge 0$ and some
  upper bound $b\ge a$: Simply add $a$ to each part of each
  composition/partition that $A(n-ka,k,0,b-a)$ outputs. 

\item \textbf{Restrictions on part \textsc{number}?} Paper
  \cite{opdyke} claims that it allows for another generality: Allowing
  to compute all compositions/partitions of $n$ with $k$ parts, each
  between $a$ and $b$, \textbf{where $k$ ranges from some $k_{\text{Min}}$ to
  some $k_{\text{Max}}$}. This may be an interesting problem, but not
  so if the solution is to apply the original algorithm $A(n,k,a,b)$
  in the form:
  \begin{align*}
    A(n,k_{\text{Min}},a,b),\: A(n,k_{\text{Min}+1},a,b),\:\ldots,\: A(n,k_{\text{Max}},a,b),
  \end{align*}
  that is, if the original algorithm $A(n,k,a,b)$ is invoked simply
  with different input arguments for the parameter $k$. Doing it in
  this way generates no additional efficiencies and is also
  independent of algorithm $A$ --- a faster algorithm would be a
  better choice for $A$ than a slower one. 

  In fact, to actually understand the dimension of the suggested
  approach of simply invoking $A(n,k,a,b)$ for different values of
  $k$, consider another trivial generalization of the mentioned
  type. We could, for example, define the \textbf{quadruply restricted
  integer composition/partition problem} of generating all integer
  compositions/partitions of $n\in N$ with parts $k\in K$, lower
  bounds $a\in A$  
  and upper bounds $b\in B$, where $N,K,A,B$ are arbitrary sets. This
  may be an interesting problem, but solving it via the algorithm
\begin{verbatim}
for n in N
  for k in K
    for a in A
       for b in B
           A(n,k,a,b)
\end{verbatim}
is simply a trivial solution that is not worth mentioning. Again,
faster algorithms $A(n,k,a,b)$ should then always be preferred over
slower algorithms $A(n,k,a,b)$. 

\item \textbf{Speed?}
As mentioned, the algorithm designed in \cite{opdyke} is slow. It
takes time $O(k)$ per composition and is therefore inefficient (see,
e.g., \cite{vajn}). 

\item \textbf{A ``generalized'' and ``unified'' approach?}
In 
mathworks comments, the author of \cite{opdyke} argued that
while his algorithm is slow, it trades this off by generality: It
solves a generalized problem with varying number of parts and
arbitrary upper and lower bounds. However, I have outlined that these
are 
not generalizations. In essence, thus, the algorithm is as specific as
any algorithm that solves the restricted integer composition/partition
problem. Only, it is so at a worse runtime. 

Besides this, the author claimed that his approach is ``unified'' in
the sense that few modifications are necessary to transform the
\emph{composition} algorithm into a \emph{partition} algorithm and
vice versa, and in that his recursions outline fundamental links
between integer compositions and partitions. In my opinion, this is a
weak argument. From a practical 
point of view, I'd rather have two fast and very distinct algorithms
than two slow ones that are very similar, wouldn't I?

From a theoretical
perspective, what is so surprising about two similar algorithms (or, recursions) for
the integer composition and partition problem? After all, compositions
\emph{are ordered partitions}, so similarities, \emph{per se}, should not come as a
surprise. 

To illustrate, note the following property of restricted integer
compositions/partitions. In writing $n=\pi_1+\cdots+\pi_k$, with each
$\pi_i\in 
[a,b]$, one can use $b$ either $0$ times, $1$ time, ..., up to
$\lfloor n/b \rfloor$ times. If one uses $b$ exactly $i$ times (as
first part(s)), one is
left with the problem of solving $n-bi=q_1+\ldots+q_{k-i}$, where
$q_1,\ldots,q_{k-i}\in [a,b-1]$. Hence, restricted integer
partitions satisfy the `recursion'
\begin{align}\label{eq:partition}
  \struct{P}_{[a,b]}(n,k) = \bigcup_{i} \struct{P}_{[a,b-1]}(n-bi,k-i).
\end{align}
If one redistributes the $i$ $b$'s among the total of $k$ parts, one
sees that restricted integer compositions satisfy the `recursion'
\begin{align}\label{eq:fast}
  \struct{C}_{[a,b]}(n,k) = \bigcup_{i} \binom{k}{i}\struct{C}_{[a,b-1]}(n-bi,k-i),
\end{align}
where, sloppy, we let $ \binom{k}{i}\struct{C}_{[a,b-1]}(n-bi,k-i)$
denote the distribution of the $i$ parts among $k$. Hence, the
following recursion formulae exist:
\begin{align*}
  p_{[a,b]}(n,k) &= \sum_i p_{[a,b-1]}(n-bi,k-i),\\
  c_{[a,b]}(n,k) &= \sum_i \binom{k}{i}c_{[a,b-1]}(n-bi,k-i).
\end{align*}
These recursions also immediately `show' a similarity between integer
compostions and partitions and they seem at least as useful to
demonstrate this relationship as are the formulas in \eqref{eq:2} and
\eqref{eq:8}. In fact, these two recursive relationships can
immediately be used for providing an algorithm for the restricted
integer composition/partition problem (incidentally, this was my
hand-coded approach 
that was faster than the algorithm in \cite{opdyke} ...). Still, this
``unifying'' principle among the two recursions is, as it seems to me,
not yet justified in suggesting, on a journal level, yet another
algorithm for the restricted integer composition/partition problem. 

\end{enumerate}

\section{Experiments}
To compare algorithms for generating restricted integer compositions,
we run the following experiments. We generate all restricted integer
compositions of $n$ with $k$ parts, each between $a$ and $b$, via two
na\"ive generation algorithms, as well as via the Opdyke algorithm
\cite{opdyke} and the 
algorithm suggested in \cite{vajn}. The two na\"ive algorithms are:
\begin{itemize}
  \item[(i)] The algorithm that generates all compositions of $n$ with
    $k$ parts, each between $a$ and $b$, by recursively generating all
    compositions of $n-x$ with $k-1$ parts and then adding $x$ to these,
    for $x\in [a,b]$. This is a direct implementation of recursion
    \eqref{eq:cAk}. Note that a na\"ive implementation of the latter
    recursion is clearly inefficient, since it computes the same
    things over and over again, as we illustrate in the generation
    tree in Figure \ref{fig:3}. 
  \item[(ii)] A na\"ive implementation of recursion \eqref{eq:fast}. 
\end{itemize}
Note that all four compared algorithms are fully general in the sense
of paper \cite{opdyke} in that they generate all restricted integer
compositions in the interval $[a,b]$ and in that they can also be
invoked with different values of the number of parts parameter $k$
(what paper \cite{opdyke} calls a `double restriction'). The na\"ive
algorithms can also easily be adapted to generate restricted integer
partitions, as outlined above. 
We abbreviate the algorithms as (V) for the algorithm suggested in
\cite{vajn}, (O) for the Opdyke algorithm, and (6) and (10) for the
algorithms based on direct implementations of recursions \eqref{eq:cAk}
and \eqref{eq:fast}, respectively. All implementations are our own Python
implementations. In the case of algorithms (V) and (O), we
directly implement the pseudo-code given in the respective papers. We
run the experiments on a $2.4$ GHz processor. 

Results are illustrated in Figures
\ref{fig:1} and \ref{fig:2}; throughout, we fix $[a,b]$ to 
$[1,7]$.\footnote{Nature of the results do not depend on $a$ and $b$.}
In Figure \ref{fig:1}, we plot run time as a function of $n$
($n\in[10,22]$), for 
$k=\frac{n}{2}$.  We see that for this middle value of parts
($k=n/2$), the ordering of algorithms in terms of run time is
(V)$<$(10)$<$(O)$<$(6),
which renders the algorithm
(O) suggested in \cite{opdyke} worst, except for the `baseline'
algorithm (6). For $n=22$ and $k=11$, algorithm (V) takes about
$0.89s$, (10) takes $1.42s$, (O) takes $3.68s$ and (10) takes
$6.29s$. 

In Figure \ref{fig:2}, we plot how our results depend on the number of
parts $k$, fixing $n$ at $n=22$. 
We
see that as $k$ is small, the algorithms (V), (O), and (10) are all
roughly equally fast (run time is in fractions of seconds), but, as
$k$ increases, the algorithms (O) and (6) become very bad. For
example, at $k=16$, run times are
\vspace{0.1cm}

\begin{tabular}{c|rr}
  (V)  & $0.10s$& $1$\\
  (10) & $0.27s$& $2.7$\\
  (O)  & $6.53s$& $65.3$\\
  (6)  & $6.15s$ & $61.5$\\
\end{tabular}

\vspace{0.15cm}
\noindent whence (O) is roughly $65$ times slower than (V). In Figure
\ref{fig:3}, we plot the relative run time of (O) in terms of the run
time of 
(V), as a function of $k$. Given the logarithmic scale of the plot, we
note that (O)'s relative run time, with respect to (V), increases
exponentially in $k$, the number of parts. 
\begin{figure}
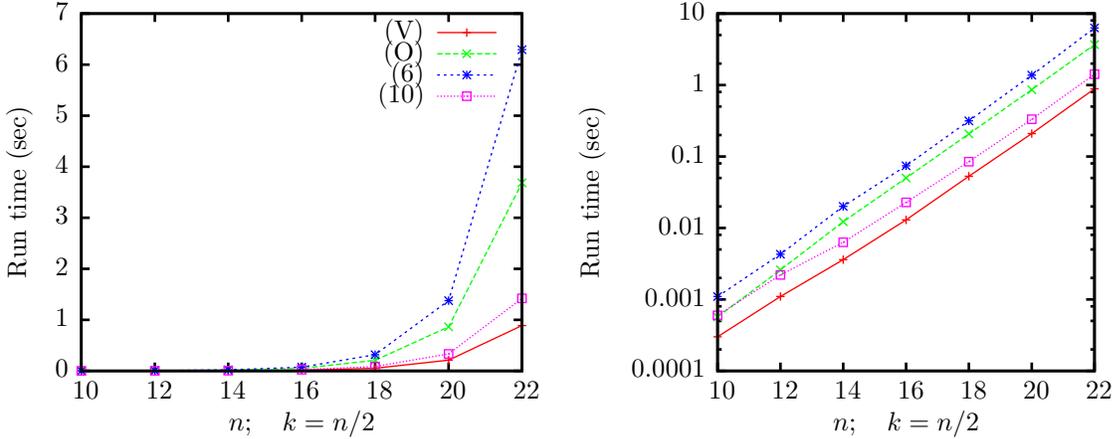

  \input{plots/comparison2_nolog.tex}
  \input{plots/comparison2_nolog_2.tex}
  \caption{Left: Run time of algorithms as a function of $n$, with
    $k=n/2$. Right: Logarithmic scale of left. Throughout: Averages
    over $10$ runs.}
  \label{fig:1}
\end{figure}
\begin{figure}
  \input{plots/comparison1_nolog.tex}
  \input{plots/comparison1_nolog_2.tex}
  \caption{Left: Run time of algorithms as a function of $k$, with
    $n=22$ fixed. Right: Logarithmic scale of left. Throughout: Averages
    over $10$ runs.}
  \label{fig:2}
\end{figure}
\begin{figure}
  \begin{center}
  \input{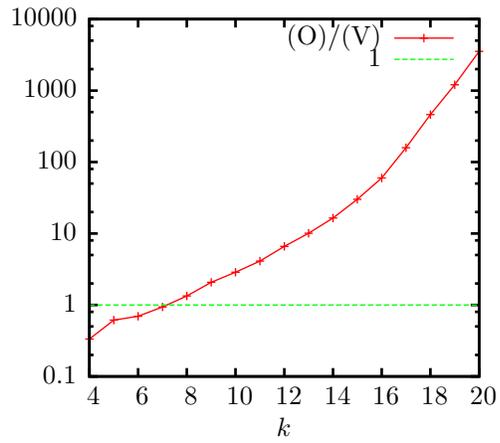}
  \end{center}
  \caption{Ratio of run time of (O) 
    and run
    time of (V), as a function of $k$, with $n$ fixed at $n=22$. The line
    $y=1$ indicates values 
    where
    (O) and (V) have the same run time.}
  \label{fig:3}
\end{figure}

Why is the algorithm designed in \cite{opdyke} so slow? In the end, it is
because it makes too many recursive calls; in particular, it
repeatedly recursively calls itself with the same input parameters, a
feature it shares with algorithm (6). To see this, if recursion
\eqref{eq:cAk} is invoked with some input parameters $n$ and $k$, then
$c_A(n,k)$ will recurse to $c_A(n-a,k-1)$, for $a\in A$. This, in
turn, will recurse to $c_A(n-a-b,k-2)$, for $b\in A$. However, the
algorithm will, in this way, recompute the value $c_A(n-a-b,k-2)$ for
all summations of $a+b$. For instance, if $A=[1,2]$, then $c_A(n,k)$
will call $c_A(n-1,k-1)$ and $c_A(n-2,k-1)$. The former will call
$c_A(n-2,k-2)$ and $\mathbf{c_A(n-3,k-2)}$, while the latter will call
$\mathbf{c_A(n-3,k-2)}$ (again!) and $c_A(n-4,k-2)$ (repetitions in
bold), and so on. Of course, this redundancy is given in each part of the
generation tree, making algorithm (6) highly inefficient. If indeed
the algorithm designed in \cite{opdyke} is based on formula
\eqref{eq:cAk}, it 
is not surprising that it is also highly inefficient. 

In Figure \ref{fig:3}, we show the generation trees of our four
outlined algorithms when invoked with input parameters $n=6$, $k=5$,
and $a=1$, $b=3$; note that $c_{[a,b]}(6,5)=5$. Overall, algorithm (V)
makes $5$ recursive calls (excluding the top node), which is
optimal. Algorithm (10) makes $12$ recursive calls. In contrast,
algorithm (O) makes $19$ recursive calls and algorithm (6) $41$. 

Finally, there is no need to also experiment with the integer
partition generation algorithm designed in \cite{opdyke}, precisely
because it is so similar to the composition generation algorithm and
accordingly shares all its inefficiencies. Algorithm (10), or rather
its partition analogue as outlined above (given in Equation
\eqref{eq:partition}), will surely be more 
efficient and the partition analogue of recursion (6) will be less
efficient. More specialized partition generation algorithms, in turn,
will be superior (even) to algorithm (10). 

\begin{figure}
  \centering
  \includegraphics[scale=0.3]{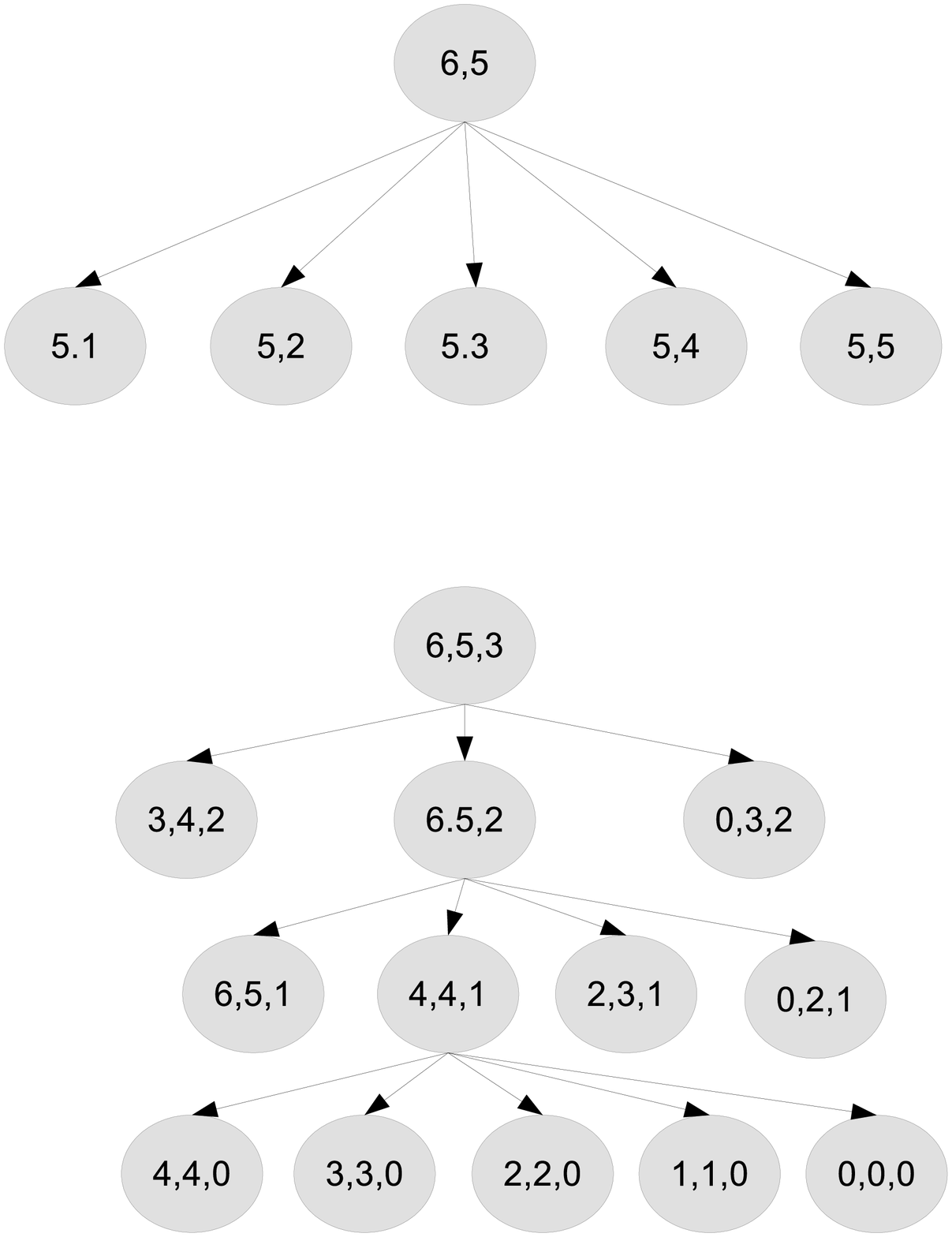}
  \includegraphics[scale=0.3]{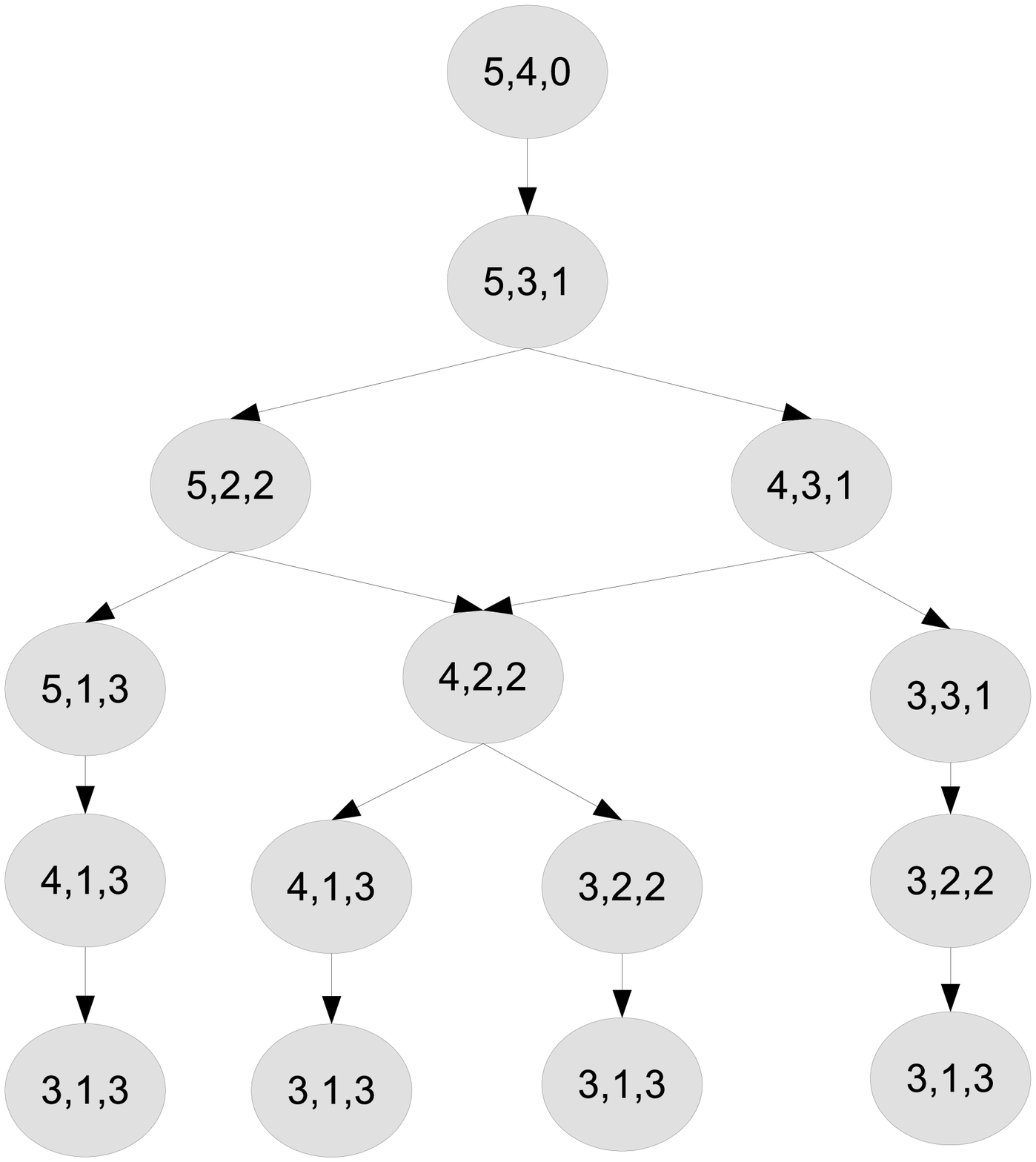}
  \includegraphics[scale=0.3]{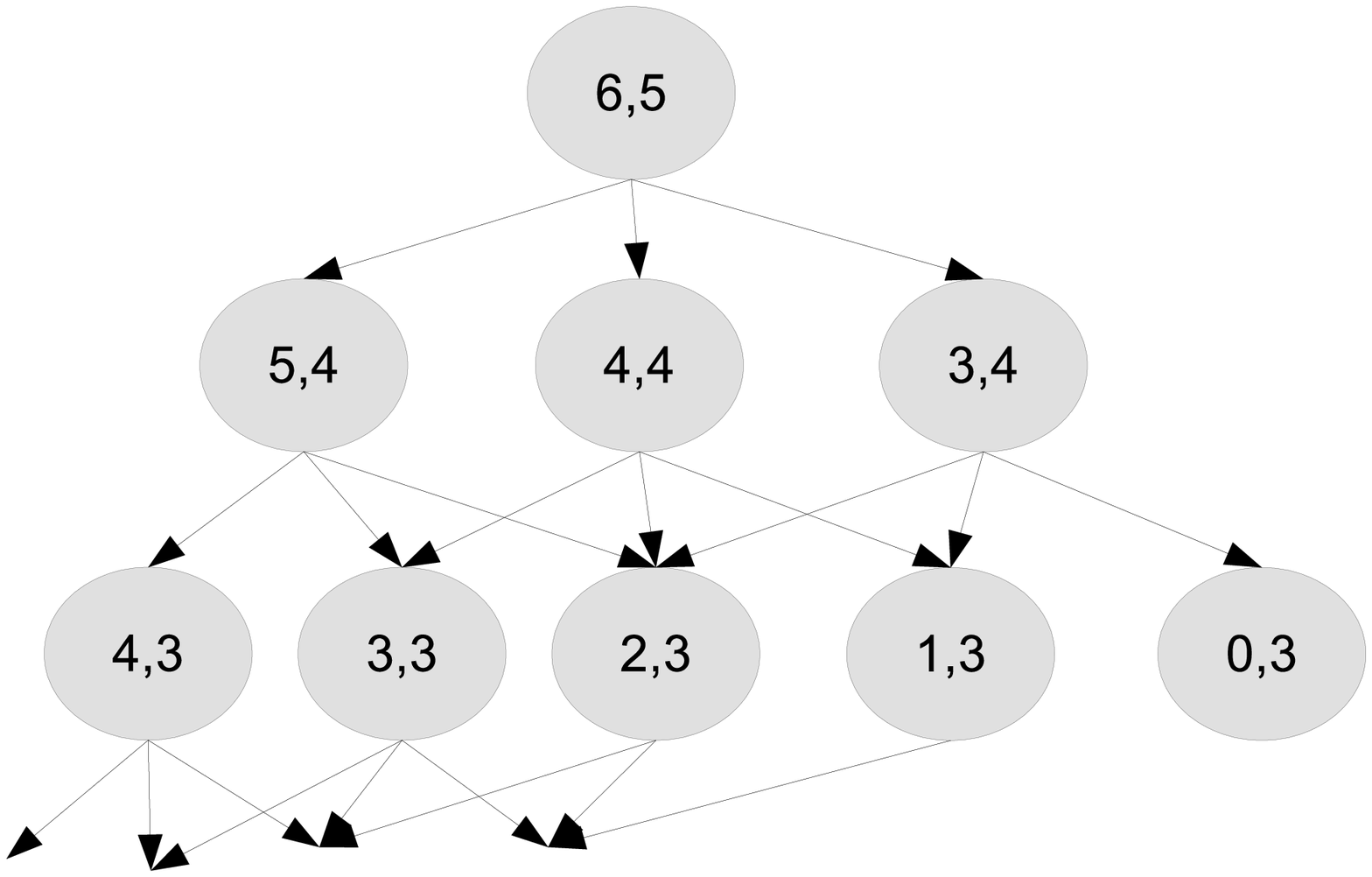}
  \caption{Generation trees induced by our four algorithms. Top left:
    Algorithm (V) (top) and algorithm (10) (bottom). Top right:
    Algorithm (O). Bottom: Top part of the tree of algorithm (6). The
    numbers in the nodes refer to input parameters of the algorithms
    in the recursive calls.}
  \label{fig:3}
\end{figure}

\bigskip

\textbf{Summary}: In my opinion, the Opdyke \emph{algorithm} is 
not general, and merely slow. 
The \emph{recursion formulae} are appealing, if
only they were not so well-known 
and simple. In my view, it has been 
unfortunate 
to describe these recursions as
`solutions to open problems', when in fact the integer composition
recursions are outlined on the second page of bibliography the author 
cites. 
Possibly, the paper 
can make other contributions to this
field in the future, but, in my opinion, these would still have to be
unveiled.\footnote{What I think the paper may be contributing is to
  outline a relationship between restricted integer compositions and
  Pascal's triangle, similar in spirit to what is shown in \cite{fielder}.}

\end{document}